 \documentclass[conference,final,twocolumn]{IEEEtran}
\usepackage{mathrsfs}
\usepackage{algorithm}
\usepackage{algorithmic}
\usepackage{graphicx}
\usepackage{cite}
\usepackage{citesort}
\usepackage{color}
\usepackage{psfrag}
\usepackage{subfigure}
\usepackage{amssymb}
\usepackage{epsfig}
\usepackage{pifont}
\usepackage{amsmath}
\usepackage{array}
\usepackage{multicol}
\usepackage{multirow}
\usepackage{pifont}
\usepackage{indentfirst} 
\usepackage{amsfonts}
\usepackage{fancyhdr}
\usepackage{amscd}
\usepackage{bm}
\usepackage{tikz}
\usepackage{xcolor}
\usepackage{etoolbox}
\usetikzlibrary{chains,fit,shapes}

\usetikzlibrary{decorations}
\usetikzlibrary{decorations.pathreplacing}
\usetikzlibrary{calc}
\usetikzlibrary{arrows}

\newtoggle{quickdecim}

\newtheorem{proposition}{Proposition}

\IEEEoverridecommandlockouts
\begin{document}
	
	\title{Role of Large Scale Channel Information on Predictive Resource Allocation}
\author{
	\authorblockN{\large{Chuting Yao and Chenyang Yang}}
	\vspace{0.2cm}
	\authorblockA{Beihang University, Beijing China\\
		Email: \{ctyao, cyyang\}@buaa.edu.cn }
	\thanks{This work was supported by National Natural
		Science Foundation of China under Grant 61120106002 and National Basic
		Research Program of China under Grant 2012CB316003.}
}

%
\maketitle

\begin{abstract}
When the future achievable rate is perfectly known, predictive resource allocation can provide high performance gain over traditional resource allocation for the traffic without stringent delay requirement. However, future channel information is hard to obtain in wireless channels, especially the small-scale fading gains. In this paper,
we analytically demonstrate that the future large-scale channel information can capture almost all the performance gain from knowing the future channel by taking an energy-saving resource allocation as an example. This result is important for practical systems, since large-scale channel gains can be easily estimated from the predicted trajectory of mobile users and radio map. Simulation results validate our analysis and illustrate the impact of the estimation errors of large-scale channel gains
on energy saving.
\end{abstract}

\section{Introduction}
As the smart phone popularizes, improving the performance of mobile networks such as energy efficiency (EE) by exploiting context information
has drawn attentions  \cite{Abou2013Predictive,abou2014toward,Abou2014Energy,Draxler2014Anticipatory,yao2015context}.

Context information can be classified into application (e.g.,
quality of service), network (e.g., congestion status), and
user (e.g., trajectory of a mobile user) levels \cite{Choongul2012concept}. The location and trajectory of a  user in the future can be predicted from analyzing the user behavior
\cite{Skog2009Intelligent}, and the bandwidth of a network can be predicted by a base station (BS) from the traffic measurements in the past \cite{Zheng2013Optimizing}. For the
traffic not having an urgent deadline for transmission, such as pre-subscribed file downloading, content pushing, and on-demand video streaming delivery, resource planning with user level context information, or predictive resource allocation, can save  energy and improve performance of a network significantly \cite{Abou2014Energy,abou2014toward,Draxler2014Anticipatory,yao2015context,Riiser2012Video,Zheng2013Optimizing}. With
perfect achievable rate prediction, i.e., assuming perfect future channel information, 
half energy can be saved \cite{abou2014toward}.

However, in wireless channels the small-scale channel gains in the future are hard to know accurately, where the channel prediction errors depend on the channel coherence time. By contrast, the large-scale channel gains can be estimated from user locations and the radio map \cite{abou2014toward}, which however has estimation errors due to the prediction error of locations and the measurement error of the signal strength as reported by \cite{Abou2015Evaluating}. 
Recently, the simulation results provided in \cite{yao2015context} demonstrate that the power-saving gain obtained from knowing the future large-scale channel gains is almost the same as that provided by perfectly knowing all the future instantaneous channel gains. This observation motivates a natural question to ask: under what condition can we only use large-scale channel information to achieve the potential of predictive resource allocation? 

In this paper, we strive to answer the question by examining the solution of an energy-saving predictive resource allocation. To reveal the essential role of the large-scale channel information, we consider a traffic with long deadline, which is modeled as transmitting a given amount of data in a long
duration, and simplify the system model in  \cite{yao2015context} to a single user scenario without background traffic. From the optimal solution of an energy minimization problem with perfect future channel information, we show that only
two key parameters in the solution, threshold and water-filling level, depend on the future channels.
By deriving the asymptotic distribution of these two parameters, we show that the threshold and water-filling level can be estimate accurately with large-scale channel gains.
Simulation results validate our analysis, and show that imperfect large-scale channel information causes minor performance degradation for predictive resource allocation.

\section{System Model}
Consider a multi-cell system, where each BS is equipped with $N_t$
antennas, and transmits  in a time-slotted fashion. The duration of each time slot is $\Delta_t$.  A single-antenna user moves across the cells, who requests to download a file with $B$ bits, which needs to be conveyed before a long deadline with duration $T\Delta_t$.

We divide the duration $T\Delta_t$ into $T_f$ frames, and each frame into $T_s$ time slots.
Hence, the duration contains $T= T_fT_s$ time slots.
The large-scale channel may vary among different frames due to user mobility. The small-scale channel is modeled as block fading, which remains constant in each time slot and may vary among time slots.
For emphasizing the role of large-scale channel information and for mathematical tractability, we assume that the user is only accessed to the closest BS, and denote $m^t\in \{1,0\}$ as the scheduling indicator. When $m^t=1$, the user is scheduled by the BS, otherwise it is not. Then, the received signal at the user in $t$th time slot is \vspace{-1mm}
\begin{equation}\label{E:signal}
\textstyle y^t =  m^t \sqrt{\alpha^{\lceil\frac{t}{T_s}\rceil}} ({\bf h}^t)^H {\bf
	w}^t\sqrt{p^t} x^t
+ n^t,
\end{equation}
where $x^t$ is the transmit symbol with $\mathbb{E}\{|x^t|^2\}=1$, $p^t$ is the transmit power, ${\bf{w}}^{t} \in \mathbb{C}^{N_t \times	1}$ is the beamforming vector, ${\bf{h}}^{t} \in \mathbb{C}^{N_t
\times 1}$ is the Rayleigh fading channel vector between the user and its closest BS with independent and identically distributed (i.i.d.) elements, $\alpha^{\lceil\frac{t}{T_s}\rceil}$ is the large-scale channel gain including path loss and shadowing, and $n^t$ is
zero-mean additive white Gaussian noise with variance $\sigma^2$.
$\mathbb E\{\cdot\}$ represents expectation, and $\lceil\cdot\rceil$ is the ceiling function.
Since single user is scheduled in each time slot, maximum ratio transmission
is optimal, i.e., ${\bf{w}}^{t} = {\bf{h}}^{t}
/\|{\bf{h}}^{t} \|$, where $\|\cdot\|$ denotes the Euclidean norm.

For notational simplicity, we consider unit bandwidth. In the $t$th time slot, the achievable rate in nats is \vspace{-1mm}
\begin{equation}
R^t = m^t \ln (1+g^t p^t),
\end{equation}
where $g^t\triangleq \alpha^{\lceil\frac{t}{T_s}\rceil} \|{\bf{h}}^{t}\|^2/\sigma^2$ is the \emph{equivalent channel gain}.

Assume that a BS can be switched into sleep mode when
the BS dose not serve the user in a time slot. The total power
consumed at the BS in the $t$th time slot can be modeled as \vspace{-1mm}
\begin{align}\label{E:PowerModel}
\textstyle p^t_{\rm tot} = \frac{1}{\xi} p^t + m^t(p_{\rm act}-p_{\rm sle}) +p_{\rm sle},
\end{align}
where  $\xi$ is the power amplifier efficiency, $p_{\rm act}$ and
$p_{\rm sle}$ are the circuit power consumptions when the BS is in active and
sleep modes, respectively.

\section{Resource Allocation with Perfect Channel Information}
To capture the essential role of large-scale channel gain in
resource allocation for conveying the $B$ bits (i.e., $B\ln2$ nats) in a long deadline, we first formulate an energy minimization problem with perfect
large-scale
and small-scale channel information in all time slots. The optimal scheduling and power allocation to minimize the total energy consumed in
the $T$ time slots can be found from the following problem,
\begin{subequations}\label{P:original}
	\begin{align}
	\min_{p^1,\ldots,p^T, m^1,\ldots, m^T} &\textstyle\sum_{t=1}^{T} p^t_{\rm tot} \Delta_t\\ \label{P:1}
	s.t. & \textstyle \sum_{t=1}^{T} m^t \ln (1+p^tg^t) = \frac{B\ln 2}{\Delta_t} \triangleq R,\\ \label{P:2}
	&p^t\geq 0,  m^t\in \{1,0\}, \quad t=1,\cdots, T,
	\end{align}
\end{subequations}
where \eqref{P:1} is the constraint on transmitting the $B$ bits within the duration $T \Delta_t$.
To simplify the analysis later, the maximal power constraint is not considered.

Problem \eqref{P:original} can be solved using similar method as in \cite{yao2015context}, which is summarized as the following two steps.

In the first step, scheduling and power allocation is optimized for a given number of active time slots $N \triangleq \sum_{t=1}^{T}m^t$. Since the circuit power consumption is given when $N$ is given, we only need to minimize the total transmit power in the $T$ time slots. Therefore, the time slots with large equivalent channel gains are selected to transmit, i.e., the scheduling indicator can be obtained as $m^t= {\bf 1}(g^t\geq g_{\rm th})$, where  $g_{\rm th}$ is a threshold and ${\bf 1 }(x)=1$ when the event $x$ is true, otherwise, ${\bf 1 }(x)=0$. The power allocation minimizing the total transmit power can be found from a standard power allocation problem, whose solution is
\begin{equation}\label{E:solution_x}
\textstyle p^{t} = \big({\nu  - \frac{1}{g^t}}\big){\bf 1}({g^t\geq \max\{g_{\rm th},\frac{1}{\nu}\}}),~t=1,\cdots, T,
\end{equation}
where $\nu$ is the water-filling level satisfying \vspace{-2mm}
\begin{align}\label{E:nu_x}
\textstyle \nu
= {\rm exp}\big(\frac{R}{L}-\frac{1}{L}{\sum_{\cal L}
		\ln g^t }\big),
\end{align}
 $L$ is the number of  time slots allocated with non-zero power among the scheduled time slots, and $\cal L$ is the set of the corresponding time slots.

\emph{Remark 1}: Since $N = \sum_{t=1}^{T}m^t=\sum_{t=1}^{T} {\bf 1}(g^t\geq g_{\rm th})$ and considering \eqref{E:solution_x}, if $\nu g_{\rm th}\geq 1$, $L = N$ time slots will be allocated with non-zero power, otherwise $L<N$.

In the second step, the  number of scheduled time slots $N$ is optimized to minimize the total energy consumption. Then, the optimal threshold $g^{*}_{\rm th}$ can be obtained by selecting $N^*$ time slots with largest equivalent channel gains, and optimal water-filling level $\nu^*$ can be obtained from \eqref{E:nu_x} by setting $L = N^*$. The optimal scheduling can be obtained as $m^{t*}= {\bf 1}(g^t\geq g^*_{\rm th})$, and the optimal transmit power $p^{t*}$ can be obtained from \eqref{E:solution_x} with the optimized water-filling level $\nu^*$.

We can observe from the optimal solution of problem \eqref{P:original} that the power allocated in the
$t$th time slot depends on the  equivalent channel gain $g^t$ in this time
slot, as well as the channel information in
other time slots implicitly included in the optimal water-filling level
$\nu^*$ and threshold $g^*_{\rm th}$. This suggests that if we can obtain the two parameters $\nu^*$ and $g^*_{\rm th}$, the explicit future channel
information in the $t+1$th, $\cdots, T$th time slots is no longer necessary.


\vspace{-1mm}
\section{Role of Large-scale Channel Information}\vspace{-1mm}
In this section, we show that the optimal water-filling level
$\nu^*$ and threshold $g^*_{\rm th}$ can be estimated accurately with the large-scale channel gains when the value of $T$ is large. This indicates that the large-scale channel information plays the key role
on the energy-saving predictive resource allocation. Specifically,
we analyze the estimation accuracy of $\nu^*$ and $g_{\rm th}^*$ when $\alpha^t,t=1,\ldots,{T_f}$ are known but
small-scale channels are unknown.

For channel vector $\bf h$ with elements subject to i.i.d. Rayleigh fading, the small-scale channel gain $\|{\bf h}^t\|^2$ follows
Gamma distribution with probability density function (PDF) as
\begin{equation}\label{E:fh}
\textstyle f_h(h) = \frac{e^{-h}h^{N_t-1}}{\Gamma(N_t)},
\end{equation}
where ${\Gamma(\cdot)}$ is the Gamma function.

The following proposition shows that the
distribution information for the equivalent channel gains in all the ${T = T_f T_s}$ time slots can be estimated with the large-scale channel gains.

\begin{proposition}\label{Pro:1}
	When $\alpha^j,j=1,\ldots,{T_f}$ are known and	$T_s \to \infty$, i.e., the small scale channels in each frame are ergodic, the set of equivalent channel gains $\{g^t=\alpha^t \|{\bf{h}}^{t}\|^2/\sigma^2, t=1, \ldots,{T}\}$ have the same elements but different orders as a set of $T$ i.i.d. random variables (denoted as $\tilde g^t,t=1,\ldots,T$) with PDF as
	\begin{equation}\label{E:PDFg}
\textstyle	f(g)=\frac{1}{T_f}
	\sum_{j=1}^{T_f}\frac{({\frac{\sigma^2}{\alpha^j}g})^{N_t -1}}{\Gamma(N_t)}
	e^{-\frac{\sigma^2}{\alpha^j}g},
	\end{equation}	
\end{proposition}
\begin{proof}
	See Appendix \ref{Proof:Pro1}.
\end{proof}
 Note that the distribution information in \eqref{E:PDFg} only depends on where the user has been, but dose not depend on the time when the user is there. 

In practical systems, the large-scale channel gains can be estimated from the radio map \cite{abou2014toward} with the help of the predicted user location, which inevitably have estimation errors.
Fortunately, the imperfect large-scale channel information has little impact on the estimated distribution information
$f(g)$ in \eqref{E:PDFg}, as demonstrated via simulations later.


\vspace{-1mm}
\subsection{Estimation Accuracy with Large-scale Channel Gains}
Since the optimal threshold $g_{\rm th}^*$ and water-filling level
$\nu^*$ are obtained from problem
\eqref{P:original}  in two steps, we first analyze the accuracy of estimating these two parameters with large-scale channel information when the number of active time slots
$N$ is given, and then analyze the accuracy of estimating $N^*$ with large-scale channel information.

\subsubsection{Estimation accuracy of $g_{\rm th}$ and $\nu$ with given
	$N$} Define $\kappa
\triangleq N/T= \sum_{t=1}^{T} {\bf 1}(g^t\geq g_{\rm th})/T $ as the active
ratio of the $T$ time slots. Then, given
$N$ is the same as given $\kappa$. Further considering that $\tilde g^t,t=1,\ldots, T$ are  $g^t, t=1,\ldots, T$ in different orders, $\kappa = \sum_{t=1}^{T} {\bf 1}(\tilde g^t\geq g_{\rm th})/T$.
\begin{itemize}
	\item {\bf Estimation  accuracy of $g_{\rm th}$ with given $\kappa$}
	
	Since $g_{\rm th}$ is the threshold to select $T\kappa$ time slots with largest equivalent
	channel gains, it is the sample $[1-\kappa]$-quantile of
	population $\tilde g^t,t=1,\ldots,T$. According to \cite{Bahadur1966Note}, when
	$T_s \to \infty$ and hence $T= T_sT_f\to\infty$, $g_{\rm th}$ follows normal
	distribution as,\vspace{-1mm}
	\begin{align}\label{E:gth_distribution}
	g_{\rm th} \sim {\mathbb N}\left(\mu_{g_{\rm th}},\sigma^2_{g_{\rm th}}\right),
	\end{align}
	where the mean value $\mu_{g_{\rm th}} \triangleq g_{[1-\kappa]}$ is the
	$[1-\kappa]$-quantile of $\tilde g^t$ (i.e., $g_{[1-\kappa]}$ satisfies $
	\int_{g_{[1-\kappa]}}^\infty f(g){\rm d} g=\kappa $), and the variance
	$\sigma^2_{g_{\rm th}}=\frac{\kappa(1-\kappa)}{Tf^2(g_{[1-\kappa]})}$.
	
	When $\kappa$ is given,
	$\lim_{T\to\infty}\sigma^2_{g_{\rm th}}=0$. Then, we have \vspace{-1mm}
	\begin{equation}\label{E:limgth}
	\lim_{T\to\infty} g_{\rm th} = \mu_{g_{\rm th}},
	\end{equation}
	which can be determined when $\kappa$ and $f(g)$ are known. This implies that the threshold can be estimated with $\alpha^j,j=1,\ldots,{T_f}$ for a given value of $N$, and the estimation errors approach zero when
	$T \to \infty$.
	\item {\bf Estimation accuracy of $\nu$ with given $\kappa$}
	
	The water-filling level in \eqref{E:nu_x} depends on $L$, hence can be expressed as different forms according to the relation between $g_{\rm th}$ and $\frac{1}{\nu}$, as indicated in \emph{Remark 1}.

	{\bf{Case} 1}: When $\nu g_{\rm th}\geq 1$, $L=N$. Further considering that $\sum_{t=1}^T {\bf1}(\tilde g^t\geq
	g_{\rm th})= N=T\kappa$, the water-filling level  in \eqref{E:nu_x} can be derived  as \vspace{-2mm}
	\begin{align}\label{E:nu}
	\textstyle \nu = & \textstyle {\rm exp}\big(\frac{R}{N}-\frac{\sum_{t=1}^T
			\ln \tilde g^t {\bf 1}(\tilde g^t\geq g_{\rm th})}{N}\big)
	\nonumber \\
	=& \textstyle {\rm exp}({\frac{R}{T\kappa}-\frac{1}{T\kappa}\sum_{i=1}^{T\kappa}\ln  g^{[i]}}),
	\end{align}
	where $g^{[i]},i= 1,\ldots,T\kappa$ are the $T\kappa$ largest equivalent
	channel gains selected by threshold $g_{\rm th}$ from $\tilde g^t,t=1,\ldots,T$.
	Therefore, the PDF of $g^{[i]}$ is the conditional
	PDF of $\tilde g^t$ when $\tilde g^t\geq g_{\rm th}$, which is
	$\tilde f(g) = f(g){\bf 1}(\tilde g^t\geq g_{\rm th})/\int_{g_{\rm th}}^\infty f(g) {\rm d} g$.

	According to \eqref{E:limgth}, when  $T\to \infty$, the
	threshold $g_{\rm th}$ equals to $\mu_{g_{\rm th}} = g_{[1-\kappa]}$, where
	$\int_{g_{[1-\kappa]}}^\infty f(g) {\rm d} g=\kappa$. Then, the asymptotic
	PDF of $g^{[i]}$ can be derived  as\vspace{-2mm}
	\begin{equation}\label{E:pdf_gi}
\textstyle	\lim\limits_{T\to\infty}\tilde f(g)= \frac{1}{\kappa}f(g){\bf
		1}(g\geq g_{[1-\kappa]}).
	\end{equation}

By deriving the mean value and variance of $\frac{1}{T\kappa}\sum_{i=1}^{T\kappa}\ln g^{[i]}$, the water-filling level follows log-normal distribution as shown in the following proposition. The proof is omitted due to the lack of space.
    \begin{proposition}\label{Pro:2}
		When $T\to \infty$ and $\nu g_{\rm th}\geq 1$,
		$\nu$ has the following asymptotic distribution as\vspace{-2mm}
		\begin{equation}\label{E:nu_lognormal}
	\textstyle	\nu \sim \ln {\mathbb N}\big({\frac{R }{T\kappa}-{\mu_{\Phi_T}}},
		\sigma^2_{\Phi_T}\big),
		\end{equation}
		where $\mu_{\Phi_T}= \int_{g_{[1-\kappa]}}^\infty \ln g f(g)/\kappa{\rm
			d}g$, and $\sigma^2_{\Phi_T} = \big(\int_{g_{[1-\kappa]}}^\infty (\ln
		g)^2 f(g)/\kappa{\rm d}g-\mu_{\Phi_T}^2\big)/(T\kappa)$.
	\end{proposition}
	
From the property
	of log-normal distribution, the mean and variance of $\nu$ can
	be respectively derived as\vspace{-2mm}
	\begin{equation}\label{E:mu_nu}
\textstyle\!\!\!\!	\mu_{\nu}\! =\!
	e^{\frac{R}{T\kappa}-\mu_{\Phi_T}+\frac{\sigma^2_{\Phi_T}}{2}},
\sigma^2_{\nu}\! = \!\!(\!e^{\sigma^2_{\Phi_T}} \!-\!
	1\!)e^{\frac{2R}{T\kappa}- 2\mu_{\Phi_T}+{\sigma^2_{\Phi_T}}}.
	\end{equation}
	
	From \eqref{E:PDFg}, it it not hard to show that the integration
	$\int_{g_{[1-\kappa]}}^\infty (\ln  g)^2 f(g)/\kappa{\rm
		d}g-\mu_{\Phi_T}^2$ is finite. Further considering that $\lim_{T\to\infty} g_{\rm th} = g_{[1-\kappa]}$
	and $\lim_{T\to\infty} \sigma^2_{\Phi_T} = 0$, we have \vspace{-2mm}
	\begin{align}\nonumber
	\textstyle \lim\limits_{T\to\infty}\mu_{\nu} = \lim\limits_{T\to\infty}
	e^{\frac{R}{T\kappa}-\mu_{\Phi_T}}
	{~~\text {and}~~} \lim\limits_{T\to\infty}  \sigma^2_{\nu} = 0.
	\end{align}
	Then, the water-filling level $\nu$ with given $\kappa$ when $T$ approaches
	infinity can be derived as\vspace{-2mm}
	\begin{align}\label{E:limitv}
	\lim_{T\to\infty} \nu = \lim_{T\to\infty}\mu_{\nu} = e^{-\mu_{\Phi_T}},
	\end{align}
	which can be determined when $\kappa$ and $f(g)$ are known.
	
	{\bf{Case} 2}: When $\nu g_{\rm th}\leq1$, less than $T\kappa$ time slots are
	allocated with non-zero power as discussed in \emph{Remark 1}, and hence $\nu$ does not change as $\kappa$ increases. In this case, \eqref{P:1} can be expressed as  $\sum_{t=1}^{T}\ln (\nu
	 g^t){\bf 1}( g^t\geq \frac{1}{\nu})=\sum_{t=1}^{T}\ln (\nu
	\tilde g^t){\bf 1}(\tilde g^t\geq \frac{1}{\nu})=R$  after substituting \eqref{E:solution_x}. Since $g_{\rm th}$ is determined by $\kappa$ but $\nu$ does not, when $\kappa$ is large such that $\nu g_{\rm th}\leq1$, $\nu$ follows the distribution in \eqref{E:nu_lognormal} with $\nu g_{\rm th}=1$.

\end{itemize}
	The analysis for both  cases imply that the water-filling level can be estimated with $\alpha^j,j=1,\ldots,{T_f}$ for a given value of $N$, and the estimation errors approach zero when
	$T \to \infty$.
\subsubsection{Estimation Accuracy of $N^*$}
Because $N^*$ is found from minimizing the total energy consumption, its estimation accuracy depends on the accuracy of estimating the total power consumption for any given $N$, i.e., given $\kappa$.


From \eqref{E:PowerModel} and \eqref{E:solution_x}, the total
power consumption per time slot when the active ratio $\kappa$ is given  can be
derived as\vspace{-2mm}
\begin{equation}\label{E:Omega}
\textstyle \Omega \triangleq \frac{\sum^T_{t=1} p^t_{\rm tot}}{T} = \frac{1}{\xi}\Psi_p +
\kappa (p_{\rm act}- p_{\rm sle}) + p_{\rm sle},
\end{equation}
where  the transmit power per time slot is\vspace{-2mm}
\begin{equation}\label{E:TxP}
\textstyle\Psi_p\triangleq \frac{1}{T}{\sum\limits_{t=1}^T p^t} = \frac{1}{T}{\sum\limits_{t=1}^T (\nu-\frac{1}{\tilde g^t})
	{\bf1}(\tilde  g^t\geq \max \{g_{\rm th},\frac{1}{\nu}\})}.
\end{equation}

\begin{itemize}
	\item{\bf Estimation of $\Psi_p$ with given $\kappa$}
	
According to the relation between $g_{\rm th}$ and $\frac{1}{\nu}$, $\Psi_p$ has different forms.

	{\bf{Case} 1}:
    When $\nu g_{\rm th} \geq 1$,
	from \emph{Remark 1}, \eqref{E:TxP} becomes
	\begin{align}\label{E:Psi_p_gv>1}
\textstyle	\Psi_p =&\textstyle\frac{1}{T}\sum_{t=1}^T \nu {\bf 1}(\tilde g^t\geq g_{\rm th})- \frac{1}{T}\sum_{t=1}^T\frac{1}{\tilde g^t} {\bf 1}(\tilde g^t\geq g_{\rm th})
\nonumber \\=&\textstyle\kappa \nu - \frac{1}{T}\sum_{i=1}^{T\kappa} \frac{1}{g^{[i]}},
	\end{align}
	where $g^{[i]},i=1,\ldots,T\kappa$ are the $N=T\kappa$ largest
	equivalent channel gains selected by the threshold $g_{\rm th}$, whose
	asymptotic PDF is in \eqref{E:pdf_gi}.

When $T\to\infty$, we can obtain the following proposition, whose proof is omitted due to the lack of space.	
\begin{proposition}\label{Pro:3}
		When $T \to \infty$ and $\nu g_{\rm th}\geq 1$,
		the mean and variance
		of $\Psi_p$ are respectively
		\begin{align}\nonumber
		\lim_{T\to\infty} \mu_{\Psi_p} = \lim_{T\to\infty} \kappa \mu_{\nu} -\kappa \mu_{g} {\text {~~and~~}} \lim_{T\to \infty} \sigma^2_{\Psi_p} =0
		\end{align}
		where $\mu_{g} = \int_{g_{[1-\kappa]}}^\infty \frac{1}{g\kappa}f(g){\rm d} g$.		
	\end{proposition}

The proposition  indicates that the transmit power $\Psi_p$ can be estimated as $\mu_{\Psi_p}$ without errors when $T\to\infty$.
	
	{\bf{Case} 2}: When $\nu g_{\rm th} \leq 1$,
	from \emph{Remark 1}, \eqref{E:TxP} becomes
	\begin{align}\nonumber
	\textstyle\Psi_p =\frac{1}{T}\sum_{t=1}^T (\nu -\frac{1}{\tilde g^t} ){\bf 1}(\tilde g^t\geq \frac{1}{\nu}),
	\end{align}
	which does not depend on $g_{\rm th}$. Because the water-filling level $\nu$ does not depend on $\kappa$ in this case, $\Psi_p$ also does not depend on $\kappa$. Hence, the mean and variance of $\Psi_p$ are the same as those shown in Proposition 3 with $\nu g_{\rm th}= 1$.

\end{itemize}
 The analysis implies that $\Psi_p$ can be estimated accurately as its mean value  with $\alpha^j,j=1,\ldots,{T_f}$ for any given value of $\kappa$, and the estimation errors approach zero when
 $T \to \infty$.

When $\kappa$ is given, the circuit power $\kappa(p_{\rm act}-p_{\rm sle})+p_{\rm sle}$ in \eqref{E:Omega}  is fixed. Hence, when $T$ is large, the total power consumption per time slot can be estimated accurately with its mean value $\mu_{\Omega}$ as
\begin{equation}\label{E:limit_Omega}
\textstyle \lim\limits_{T\to\infty}\Omega =\lim\limits_{T\to\infty}\mu_{\Omega} = \lim\limits_{T\to\infty}\frac{1}{\xi}\mu_{\Psi_p}+\kappa(p_{\rm act}-p_{\rm sle})+p_{\rm sle}.
\end{equation}

Since when $\kappa$ increases, more time slots are employed to convey the $B$ bits, less transmit power needs to be used in each time slot. This means that
$
\textstyle\frac{\partial \Psi_p}{\partial \kappa}\leq 0
$. Recalling that $\Psi_p$ can be estimated as $\mu_{\Psi_p}$  when $T$ is large, this indicates that
$\mu_{\Psi_p}$ is a decreasing function of $\kappa$.
Further considering that the second term of \eqref{E:limit_Omega} is an increasing function of $\kappa$, the optimal active ratio $\kappa^* = N^*/T$ can be found from $\frac{\partial \Omega}{\partial \kappa }|_{\kappa=\kappa^*}= 0$.
With the accurately estimated value of $\Omega$, the optimal number of active time slots $N^*$ can be estimated accurately with $\alpha^j,j=1,\ldots,{T_f}$ when $T$ is large.

\subsection{Impact of Not-so-long Deadline}
When the value of $T$ is finite, simply estimating $\nu^*$ and $g_{\rm th}^*$ as their mean values are not accurate. Intuitively, if the estimated water-filling
level is less than $\nu^*$ or the estimated threshold  is larger than $g_{\rm th}^*$, the $B$ bits can not be conveyed within the $T$ time slots. To transmit the $B$ bits before the deadline with high probability, we can estimate $\nu^*$ and $g_{\rm th}^*$ in the following way.


For a normal distributed random variable,  $97.5\%$ of its values are less than two standard deviations
from its mean value. Considering that the threshold asymptotically follows normal
distribution and the water-filling level
asymptotically follows log-normal distribution, we have\vspace{-2mm}
\begin{align}\label{E:estimated_nugth}
&\textstyle\hat g_{\rm th}^*
\!\triangleq\!\mu_{g_{\rm th}}\!\!-\!\!2\sigma_{g_{\rm th}}|_{\kappa=\kappa^*}\!\! =\!\! g_{[1-\kappa^*]}\!\!-\!\!2\left(\!\frac{\kappa^*(1-\kappa^*)}{Tf^2(g_{[1-\kappa^*]})}\!\right)^{\frac{1}{2}}\!\!\overset{97.5\%}{\leq}\!\!g_{\rm th}^*,\nonumber\\
&\hat \nu^*  \triangleq e^{\frac{R}{T\kappa}-\mu_{\Phi_T}+ 2\sigma_{\Phi_T}}|_{\kappa=\kappa^*}\overset{97.5\%}{\geq}\nu^*,
\end{align}
where $\overset{97.5\%}{\geq}$ denotes that the  probability of the inequality being true is larger than $97.5\%$.
By using the threshold and water-filling level being estimated with only large-scale channel gains in this way,  the $B$ bits can be transmitted within the $T$ time slots with high probability no less than $(97.5\%)^2 =95.06\% $.

\section{Numerical and Simulation Results}

In this section, we first validate the analyses and then evaluate the
energy consumption of the resource allocation with perfect  and
imperfect large-scale channel information.

We consider a multi-cell system with cell radius $D = 250$ m, $N_t=4$. A
mobile user with speed uniformly distributed in $(0,20)$ m/s requests $B=2$ GBits in $T_f=120$ frames.
Each frame contains
$T_s=100$ (or 1000) time slots each with duration $\Delta_t =10$ (or 1) ms, i.e., the duration of a frame is 1 s.
The maximal transmit
power is $40$ W, the bandwidth is $10$ MHz, and $\sigma^2 =-95 $ dBm.  The path loss
model is $35.3+37.6\log_{10}(d)$, where $d$ is the distance between the BS
and user in meter \cite{TR36.814}. The circuit
power consumption parameters are $p_{\rm act}$ = $233.2$ W, $p_{\rm sle}$
=$150$ W, and $\xi=21.3\%$, which are for a macro BS \cite{Auer2011}. The results are obtained from $1000$ Monte Carlo trails, where the moving trajectory stays the same but the small-scale fading channel is subject to i.i.d. Rayleigh block fading.
Unless otherwise specified, this simulation setup is
used for all results.
\subsection{Validation of the Analysis}
We first validate proposition 1. The real and estimated trajectory of the mobile user are shown in Fig.
\ref{F:trajectorya}. To model a mobile user moving along a road, the real trajectory is generated as a straight line whose minimum distance from the BSs is $150$ m. To model the behavior of a mobile user who may frequently change lanes during the $T$ time slots, the estimated trajectory is generated as a cosine function with amplitude $A_d=5$ m and cycle $\pi=3.14$ seconds. The large-scale channel gains are computed from the distance $d$ with the assumed path loss model. Then, the value of $A_d$ can reflect the estimation errors of the user location, which leads to the estimation errors of the large-scale channel gains. The PDFs computed from \eqref{E:PDFg} with the accurate and estimated large-scale channel gains are shown in Fig.
\ref{F:trajectoryb}, where
the PDF obtained from estimated large-scale fading gains is
with legend ``EST". We also provide the histogram of the estimated large-scale fading
gains obtained from simulation. The results indicate that  the channel distribution $f(g)$ can be obtained from
large-scale channel gains accurately even when the large-scale information are imperfect.

\begin{figure} 
	\centering
	\subfigure[ True and estimated trajectories  and the corresponding distances]{
		\label{F:trajectorya} 
		\includegraphics[width=0.42\textwidth]{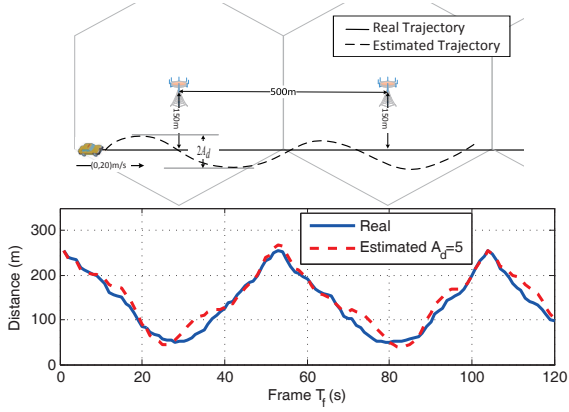}}
	\hspace{0.1cm}
	\subfigure[PDF of the equivalent channel gains]{
		\label{F:trajectoryb} 
		\includegraphics[width=0.42\textwidth]{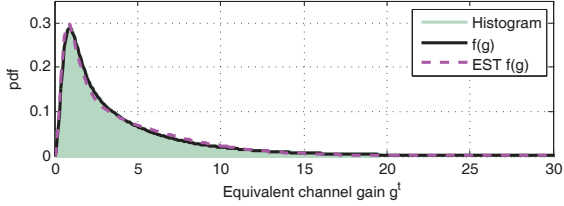}}
	\caption{PDF of $g^t$ with accurate and estimated large-scale fading gains for the true and estimated trajectories, $T_s=100$. }
	\label{F:trajectory} 
	\vspace{-4mm}
\end{figure}

To validate  \eqref{E:gth_distribution}  and proposition \ref{Pro:2}, we simulate the mean values of threshold $g_{\rm
	th}$ and water-filling level $\nu$ with different $\kappa$, and compare with $\mu_{g_{\rm th}}$ and $\mu_{\nu}$
numerically
 obtained from \eqref{E:gth_distribution} and \eqref{E:mu_nu} in Fig. \ref{F:1}(a) and
\ref{F:1}(b). We also provide the simulated standard deviations of the estimated
threshold and water-filling level when the numbers of time slots in each frame are respectively $T_s =100$ and $T_s
=1000$, as shown with the blue and green curves in the
magnified window. In Table \ref{T:1}, we further provide the deviation of $g^*_{\rm th}$ from
$\mu_{g^*_{\rm th}}$ and $\nu^*$ from
$\mu_{\nu^*}$. It can be seen that the deviation from the mean
value is small. Considering that in practice the large-scale channel gains vary in the scale of second and the small-scale channel gains change in the scale of milliseconds, this result  indicates that the threshold and water-filling level can be estimated with the large-scale channel gains accurately.

\begin{figure}
	\centering
	\includegraphics[width=0.41\textwidth]{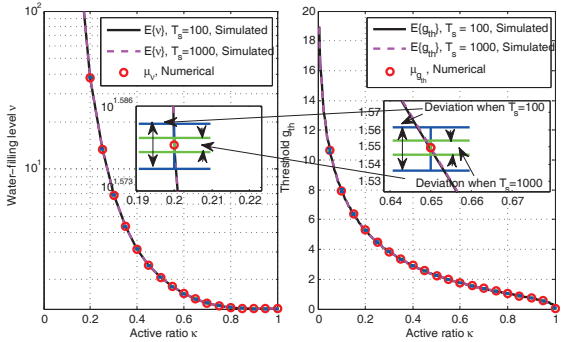}
	\caption{Simulated and numerically obtained threshold and water-filling level.}
	\label{F:1}
\end{figure}

To validate proposition \ref{Pro:3}, we simulate the mean value of the
transmit power per time slot  $\Psi_p$ and the total power
consumption per time slot  $\Omega$, and compare with analytical results $\mu_{\Psi_p}$ and $\mu_{\Omega}$ in Fig. \ref{F:2}. We can see that the analytical results
perfectly match the simulated results. In the magnified window, we show the deviation of $\Psi_p$ from its mean value when given $\kappa$, which is very
small and decreases when $T_s$ increases. This suggests that $\Omega$ can be estimated with large-scale channel gains accurately. Moreover, $\Omega$ first decreases and then increases with $\kappa$, which validates that there exists optimal active
ratio $\kappa^*$ minimizing the total power consumption and can be found by
setting $\partial \Omega/\partial \kappa |_{\kappa=\kappa^*}= 0$.

\begin{figure}
	\centering
	\includegraphics[width=0.43\textwidth]{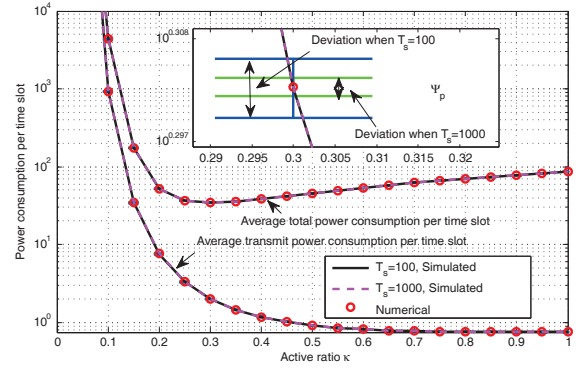}
	\caption{Power consumption per time slot.}
	\label{F:2}
	\vspace{-6mm}
\end{figure}

\begin{table}[!hbp]
	\centering \caption{Deviation of the Optimal Water-filling Level and Threshold from Mean Value Estimation}\label{T:1}
	\begin{tabular}{c|c|c}
		\hline	
		\hline
		Deviation	  & $T_s = 100$  & $T_s = 1000$ \\
		\hline
		${|\nu^*-\mu_{\nu^*}|}/{\mu_{\nu^*}}$ & $<1\%$ & $<0.3\%$  \\
		\hline
		${|g^*_{\rm th}-\mu_{g^*_{\rm th}}|}/{\mu_{g^*_{\rm th}}}$ & $<3\% $& $<1\%$\\
		\hline
		\hline
	\end{tabular}
	\vspace{-5mm}
\end{table}

\subsection{ Evaluation of the Energy Consumption}

To show the impact of only using future large-scale channel gains on the energy-saving predictive resource allocation, we have simulated the following methods.
\begin{itemize}
	\item \emph{SE-maximizing only with $g^t$} (with legend
	``SE"): The closest BS serves the user with the maximal transmit power, which can maximize the spectrum efficiency (SE) in each time slot \cite{yao2015context}.
	\item \emph{EE-maximizing  only with $g^t$} (with legend
	``EE"): The closest BS serves the user with the optimized transmit power to maximize the EE in each time slot \cite{yao2015context}.
	\item \emph{Power allocation with perfect future channel information} (with legend
	``UB"): The closest BS allocates transmit power using \eqref{E:solution_x} with $\nu^*$ and $g_{\rm th}^*$, which consumes minimal consumption to convey the $B$ bits before the deadline with duration $T\Delta_t$.
	\item \emph{Power allocation with future large-scale channel
		information} (with legend ``$A_d$"): Considering that $T_s$ is finite in the simulation, we use the conservative way to estimate the water-filling level and threshold with the large-scale channel gains in order to complete the transmission of the $B$ bits  during the $T$ time slots. Specifically, when $A_d=0$, the BS allocates transmit power by using the estimated
	water-filling level and threshold in \eqref{E:estimated_nugth} with perfect future large-scale channel
	information. When $A_d=5$ and $A_d=10$, the BS allocates power with \eqref{E:estimated_nugth} by using the estimated future large-scale channel
	information.
\end{itemize}
If the $B$ bits can not be transmitted before the deadline with duration  $T\Delta_t$, the remaining bits
will be transmitted with the maximal transmit power.

In Fig. \ref{F:5}, we provide the energy consumed by different methods for transmitting the $B$ bits during the $T$ time slots. It can be seen that knowing the accurate large-scale channel information can achieve almost the same performance as knowing all the future channel information. Inaccurate large-scale channel information leads to more energy to transmit the $B$ bits, as shown in the magnified window, but the increased energy is not significant. Without exploiting the future large-scale channel information, the ``SE'' and ``EE'' methods consume much more energy to transmit the $B$ bits in duration $T\Delta_t$. Again, this validates that the large-scale channel information plays the key role
on the energy-saving predictive resource allocation.

\begin{figure}
\centering
\includegraphics[width=0.43\textwidth]{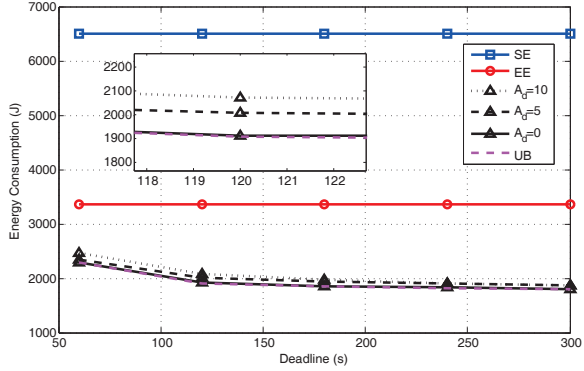}
\caption{Energy consumption vs. deadline, $T_s=100$.
}
\label{F:5}
\vspace{-6mm}
\end{figure}

\section{Conclusion}
\vspace{-2mm}
In this paper, we strived to show that large-scale channel information plays the key role on predictive resource allocation for delay tolerant services. Toward this goal, we formulated a time slot scheduling and power allocation problem minimizing the total energy consumed for conveying a given number of bits before a long deadline, where the large and small-scale channel gains are perfectly known for all time slots. We showed that only the water-filling level and threshold in the optimal solution contains future information, which can be estimated accurately with large-scale channel gains when the small-scale channels in each frame are ergodic.
Simulation results validated our analysis and showed that the estimation errors of the large-scale channel gains have little impact on energy saving.

\vspace{-2mm}
\appendices
\section{Proof of proposition \ref{Pro:1}}\label{Proof:Pro1}
Without loss of generality, we assume that the values of large-scale fading
gains $\alpha^j, j=1,\ldots,T_f$ are within a range of
$[\alpha_{\min},\alpha_{\max}]$. We divided the range into $Q$ intervals as
$\delta_1,\ldots,\delta_{Q}$, where $\delta_q =
[\alpha_{\min}+(q-1)\delta,\alpha_{\min}+q\delta]$, and $\delta =
(\alpha_{\max}-\alpha_{\min})/Q$. The number of $ \alpha^j,j=1,\ldots,T_f$
whose values are within the interval $\delta_q$ is denoted as $x_q$. Since each frame
contains $T_s$ time slots with the same large-scale fading gain, the number of $ \alpha^j,j=1,\ldots,T$ whose values are within the interval $\delta_q$
is $T_sx_q$.

Consider a sequence of i.i.d random variables $\beta^n,n=1,\ldots,T_sT_f$, whose
probability mass function is ${\rm Pr}(\beta^n = \alpha^j) = \frac{1}{T_f},j
=1,\ldots,T_f$. Then, the probability that the value of $\beta^n$ is within the
interval of $\delta_q$ is ${\rm Pr}(\beta^n\in\delta_q)=x_q/T_f$. When $T_s \to
\infty$, based on the Borel's law of large numbers, the number of $
\beta^n$ within the interval of $\delta_q$ (denoted as $z_q$) approaches its average number, i.e., $\lim_{T_sT_f\to\infty} z_q = {\mathbb E}\{z_q\} =
{\rm Pr}(\beta^n\in \delta_q) T_sT_f
= T_s x_q$. Therefore, the set $\{\beta^n,n=1,\ldots,T_sT_f\}$ has the same
elements as the large-scale fading gains of all the time slots but with different orders.

Further considering that the small-scale fading gains $\|{\bf h}^t\|^2,t = 1\ldots,T$ are i.i.d., the set of equivalent
channel gains $\{g^t = \alpha^{\lceil\frac{t}{T_s}\rceil}
\|{\bf h}^t\|^2 /\sigma^2,t=1,\ldots,T\}$ are the same as $\{\beta^n\|{\bf h}^t\|^2
/\sigma^2, n =1\ldots,T,t = 1\ldots,T \}$. Therefore,
the PDF of $g\triangleq\beta^n\|{\bf h}^t\|^2
/\sigma^2$ can be derived as\vspace{-2mm}
\begin{align}\label{E:Pr}
\!\!\!\!f(g) = \textstyle \lim\limits_{\Delta\to 0}\frac{{\rm Pr}(g<\frac{\beta^n\|{\bf h}^t\|^2
	}{\sigma^2}\leq g+\Delta)}{\Delta}.
\end{align}
Because $\|{\bf h}^t\|^2$ follows Gamma distribution
with PDF in \eqref{E:fh}, \eqref{E:Pr} can be
further derived as\vspace{-2mm}
\begin{align}\label{E:Pr_geq}
 &\textstyle \lim\limits_{\Delta\to 0}\frac{ {\rm Pr}(g<\frac{\beta^n\|{\bf h}^t\|^2}{\sigma^2}\leq g+\Delta)}{\Delta}
\nonumber\\ =&\textstyle\lim\limits_{\Delta\to 0}\frac{\sum_{j=1}^{T_f} {\rm Pr}(\beta^n  = \alpha^j )
	{\rm Pr}({ \frac{\sigma^2}{\alpha^j}g}<\|{\bf h}^t\|^2\leq { \frac{\sigma^2}{\alpha^j}(g+\Delta)} )}{\Delta}
\nonumber\\ = &\textstyle  \frac{1}{T_f} \sum_{j=1}^{T_f}f_h(\frac{\sigma^2}{\alpha^j}g ),
\end{align}
which can be expressed as \eqref{E:PDFg}.
\vspace{-2mm}
\bibliographystyle{IEEEbib}
\bibliography{IEEEabrv,2015YCT}

\begin{thebibliography}{10}
\providecommand{\url}[1]{#1}
\csname url@samestyle\endcsname
\providecommand{\newblock}{\relax}
\providecommand{\bibinfo}[2]{#2}
\providecommand{\BIBentrySTDinterwordspacing}{\spaceskip=0pt\relax}
\providecommand{\BIBentryALTinterwordstretchfactor}{4}
\providecommand{\BIBentryALTinterwordspacing}{\spaceskip=\fontdimen2\font plus
\BIBentryALTinterwordstretchfactor\fontdimen3\font minus
  \fontdimen4\font\relax}
\providecommand{\BIBforeignlanguage}[2]{{%
\expandafter\ifx\csname l@#1\endcsname\relax
\typeout{** WARNING: IEEEtran.bst: No hyphenation pattern has been}%
\typeout{** loaded for the language `#1'. Using the pattern for}%
\typeout{** the default language instead.}%
\else
\language=\csname l@#1\endcsname
\fi
#2}}
\providecommand{\BIBdecl}{\relax}
\BIBdecl

\bibitem{Abou2013Predictive}
H.~Abou-zeid and H.~Hassanein, ``Predictive green wireless access: exploiting
  mobility and application information,'' \emph{IEEE Wireless Commun.},
  vol.~20, no.~5, pp. 92--99, Oct. 2013.

\bibitem{abou2014toward}
H.~Abou-Zeid and H.~S. Hassanein, ``Toward green media delivery: location-aware
  opportunities and approaches,'' \emph{IEEE Wireless Commun.}, vol.~21, no.~4,
  pp. 38--46, Aug. 2014.

\bibitem{Abou2014Energy}
H.~Abou-zeid, H.~Hassanein, and S.~Valentin, ``Energy-efficient adaptive video
  transmission: Exploiting rate predictions in wireless networks,'' \emph{IEEE
  Trans. Veh. Technol.}, vol.~63, no.~5, pp. 2013--2026, June 2014.

\bibitem{Draxler2014Anticipatory}
M.~Dr{\"a}xler, P.~Dreimann, and H.~Karl, ``Anticipatory power cycling of
  mobile network equipment for high demand multimedia traffic,'' in \emph{IEEE
  GREENCOM}, 2014.

\bibitem{yao2015context}
C.~Yao, C.~Yang, and Z.~Xiong, ``Power-saving resource allocation by exploiting
  the context information,'' in \emph{IEEE PIMRC}, 2015.

\bibitem{Choongul2012concept}
C.~Park, Y.~Seo, K.~Park, and Y.~Lee, ``The concept and realization of
  context-based content delivery of {NGSON},'' \emph{IEEE Commun. Mag.},
  vol.~50, no.~1, pp. 74--81, Jan. 2012.

\bibitem{Skog2009Intelligent}
I.~Skog and P.~Handel, ``In-car positioning and navigation technologies--a
  survey,'' \emph{IEEE Trans. Intell. Transportation Sys.}, vol.~10, no.~1, pp.
  4--21, March 2009.

\bibitem{Zheng2013Optimizing}
Z.~Lu and G.~de~Veciana, ``Optimizing stored video delivery for mobile
  networks: The value of knowing the future,'' in \emph{IEEE INFOCOM}, April
  2013, pp. 2706--2714.

\bibitem{Riiser2012Video}
H.~Riiser, T.~Endestad, P.~Vigmostad, C.~Griwodz, and P.~Halvorsen, ``Video
  streaming using a location-based bandwidth-lookup service for bitrate
  planning,'' \emph{ACM Trans. Multimedia Comput. Commun. Appl.}, vol.~8,
  no.~3, pp. 24:1--24:19, Aug. 2012.

\bibitem{Abou2015Evaluating}
H.~Abou-zeid, H.~Hassanein, Z.~Tanveer, and N.~AbuAli, ``Evaluating mobile
  signal and location predictability along public transportation routes,'' in
  \emph{IEEE WCNC}, 2015.

\bibitem{Bahadur1966Note}
R.~R. Bahadur, ``A note on quantiles in large samples,'' \emph{The Annals of
  Mathematical Statistics}, vol.~37, no.~3, pp. pp. 577--580, June 1966.

\bibitem{TR36.814}
{TR 36.814 V1.2.0}, ``{F}urther {A}dvancements for {E-UTRA} {P}hysical {L}ayer
  {A}spects ({R}elease 9),'' \emph{3GPP}, June 2009.

\bibitem{Auer2011}
G.~Auer, V.~Giannini, C.~Desset, and e.~I.~Godor, ``How much energy is needed
  to run a wireless network?'' \emph{IEEE Wireless Commun.}, vol.~18, no.~5,
  pp. 40--49, Oct. 2011.

\end{thebibliography}
\end{document}